\documentclass[11pt]{article}

\usepackage[T1]{fontenc}
\usepackage[utf8]{inputenc}
\IfFileExists{lmodern.sty}{\usepackage{lmodern}}{\usepackage{ae,aecompl}}

\usepackage[expansion=false]{microtype}
\usepackage[margin=1.15in]{geometry}
\usepackage{amsmath,amssymb,amsthm}
\usepackage{mathtools}
\usepackage{booktabs}
\usepackage{array}
\usepackage{enumitem}
\usepackage{listings}
\usepackage{xcolor}
\usepackage[colorlinks=true,linkcolor=blue!55!black,citecolor=blue!55!black,urlcolor=blue!55!black]{hyperref}
\usepackage[capitalize,nameinlink]{cleveref}

\newtheorem{theorem}{Theorem}[section]
\newtheorem{lemma}[theorem]{Lemma}

\newtheorem{conjecture}[theorem]{Conjecture}
\theoremstyle{definition}

\theoremstyle{remark}

\definecolor{leankw}{RGB}{0,32,128}
\definecolor{leancm}{RGB}{96,96,96}
\definecolor{leanbg}{RGB}{248,248,246}
\lstdefinelanguage{lean}{
  keywords={theorem,lemma,def,abbrev,structure,instance,namespace,end,import,
    open,by,fun,match,with,if,then,else,let,intro,exact,have,show,set_option,in,
    where,deriving,do,return,forall,exists},
  sensitive=true,
  comment=[l]{--},
  morecomment=[s]{/-}{-/},
  string=[b]",
}
\newcommand{\NN}{\mathbb{N}}
\newcommand{\ZZ}{\mathbb{Z}}
\newcommand{\QQ}{\mathbb{Q}}
\newcommand{\code}[1]{\texttt{#1}}
\newcommand{\sigmaone}{\sigma_1}

\title{\bfseries Kernel-Checked Exclusions for the\\
Erd\H{o}s--Selfridge Odd Covering Problem:\\
Any Odd Covering of $\ZZ$ Has lcm Exceeding $10000$}

\author{Ibrahim Mian \and Shayaan Siddique\\[2pt]
  \normalsize Millennium Research\\
  \normalsize\texttt{\{ibby,shayaan\}@millenniumresearch.ai}\\
  \normalsize\texttt{ibrahimnmian@gmail.com}, \texttt{shayaansiddique02@gmail.com}}

\date{July 2026}

\begin{document}
\maketitle

\begin{abstract}
The Erd\H{o}s--Selfridge odd covering problem (Erd\H{o}s problem \#7) asks
whether a covering system of $\ZZ$ exists whose moduli are all odd, distinct,
and greater than $1$. The problem is open. We present a Lean~4 formalization,
checked end to end by the proof kernel, of the exclusion: any covering of
$\ZZ$ by finitely many congruence classes with distinct odd moduli $> 1$ has
lcm of the moduli exceeding $10000$. The proof composes a formalized density
argument (a covering by divisors of $N$ exceeding $1$ forces
$2N \le \sigmaone(N)$, so the lcm is abundant or perfect), a kernel-checked
abundancy floor (no odd $N < 945$ qualifies), a family of Chinese-Remainder
capacity certificates --- decidable per-$N$ arithmetic inequalities each
refuting every covering with distinct moduli $> 1$ dividing that $N$ --- for
all $23$ odd abundant numbers below $10^4$, and a kernel-checked enumeration
establishing that those $23$ are the only odd non-deficient candidates. The
result is transported to the official \code{StrictCoveringSystem} $\ZZ$
formulation of Erd\H{o}s \#7 in \code{google-deepmind/formal-conjectures},
with a bidirectional periodicity bridge between coverings of $\ZZ$ and finite
checks over $\ZZ/N\ZZ$ suitable for consuming future SAT-style search output.
All $63$ published theorems depend on exactly $\{\code{propext},
\code{Classical.choice}, \code{Quot.sound}\}$: no \code{sorry}, no
\code{native\_decide}, no solver in the trusted base. The mathematical content
is known --- the density argument is folklore, and far larger uncertified
classifications of covering numbers exist --- so the contribution is epistemic
rather than mathematical: these exclusions are theorems of the Lean kernel,
with an axiom gate enforced mechanically in continuous integration.
\end{abstract}

% ===========================================================================
\section{Introduction}\label{sec:intro}

A \emph{covering system} of the integers, introduced by Erd\H{o}s
\cite{Erdos1950}, is a finite family of congruence classes $a_i \pmod{n_i}$
whose union is all of $\ZZ$. The classic example with distinct moduli greater
than $1$ is
\[
0 \pmod 2,\quad 0 \pmod 3,\quad 1 \pmod 4,\quad 5 \pmod 6,\quad 7 \pmod{12},
\]
of period $12$. Every known covering system with distinct moduli $> 1$ uses
at least one even modulus, and Erd\H{o}s and Selfridge famously asked whether
that is forced:

\begin{conjecture}[Erd\H{o}s--Selfridge odd covering problem; Erd\H{o}s \#7
\cite{ErdosProblems7}]
Does there exist a covering system of $\ZZ$ whose moduli are odd, distinct,
and greater than $1$?
\end{conjecture}

The problem is open in both directions and has substantial connections
elsewhere in number theory; Schinzel \cite{Schinzel1967} showed that an odd
covering system would have consequences for the reducibility of trinomials.
Deep partial results exist: Hough's resolution of the minimum modulus problem
\cite{Hough2015} bounds the least modulus of any covering system with distinct
moduli, Balister, Bollob\'as, Morris, Sahasrabudhe, and Tiba sharpened the
method dramatically \cite{BBMST2022a} and proved that no covering system with
distinct moduli has all moduli odd and squarefree \cite{BBMST2022b}, and McNew
and Setty \cite{McNewSetty2025} classified \emph{covering numbers} --- the $N$
for which some covering system uses distinct moduli $> 1$ dividing $N$ --- up
to $10^6$, using an optimization-solver pipeline.

This paper reports a formalization, in Lean~4 \cite{Lean4} over
\code{mathlib} \cite{mathlib}, of the elementary exclusion tier of this
landscape, carried out so that nothing in the evidence chain rests on
unverified computation or unformalized literature. The headline theorem,
stated exactly as in the development, is:

\begin{lstlisting}
/-- Any covering of ℤ by finitely many congruence classes with distinct
    odd moduli > 1 has lcm > 10000. -/
theorem odd_covering_lcm_gt_10000
    {ι : Type} [Fintype ι] (n : ι → ℕ) (a : ι → ℤ)
    (hgt : ∀ i, 1 < n i) (hodd : ∀ i, Odd (n i))
    (hinj : Function.Injective n)
    (hcov : ∀ x : ℤ, ∃ i, (n i : ℤ) ∣ (x - a i)) :
    10000 < Finset.univ.lcm n
\end{lstlisting}

together with its transport to the official formal statement of Erd\H{o}s \#7
maintained in \code{google-deepmind/formal-conjectures}
\cite{FormalConjectures} (\cref{sec:bridge}). Every published theorem in the
development --- $63$ in a curated manifest, and every theorem of every module
in an automated audit --- depends on exactly the three axioms of Lean's
standard classical foundation, $\{\code{propext}, \code{Classical.choice},
\code{Quot.sound}\}$, with no \code{sorry} and no \code{native\_decide}.

\paragraph{What is new here, and what is not.}
We state this as plainly as the repository's README does. The mathematical
content is known: the density/abundancy argument is folklore, and
McNew--Setty's solver-based classification reaches covering numbers up to
$10^6$, far beyond our range. The contribution is epistemic, not mathematical:
these exclusions are theorems of the Lean kernel, with no solver in the
trusted base and no appeal to unformalized literature. This is not progress on
the open question. What we hope it contributes methodologically is (i) a
formal statement of the density and Chinese-Remainder capacity machinery for
covering systems, reusable for larger certified classifications; (ii) a worked
pattern for turning per-$N$ exclusions into decidable certificates discharged
by \code{decide}; and (iii) a bidirectional bridge between coverings of $\ZZ$
and finite checks over $\ZZ/N\ZZ$ shaped exactly like the output of a SAT or
exhaustive search, so that future machine searches --- in either direction ---
can be consumed by the kernel.

\paragraph{Contributions.}
\begin{enumerate}[leftmargin=2em]
\item The headline exclusion above, plus the intermediate rungs
$\mathrm{lcm} \ge 945$ and $\mathrm{lcm} > 945$, all kernel-checked
(\cref{sec:proof}).
\item A formalized \emph{capacity certificate}: a decidable per-$N$
inequality, derived from a CRT counting bound, that refutes every covering
with distinct moduli $> 1$ dividing $N$ --- with a proved monotone-relaxation
lemma making one certificate over a full coprime family apply to every
subfamily a covering might actually use (\cref{sec:capacity}).
\item A kernel-checked enumeration that the $23$ odd abundant numbers below
$10^4$ are the only odd non-deficient candidates, via a constant-fuel
paired-divisor $\sigmaone$ evaluator engineered for kernel reduction
(\cref{sec:enumeration}).
\item The transport of the exclusion to the official
\code{StrictCoveringSystem} vocabulary over $\ZZ$ of the
\code{formal-conjectures} repository, including the completeness direction ---
every abstract odd strict covering system arises from concrete data of the
enumerable form --- so the exclusions rule out \emph{their} statement, not
merely concretely presented families (\cref{sec:bridge}).
\item An axiom-gate methodology --- curated manifest, automated whole-library
audit, continuous integration --- together with anti-vacuity and soundness
controls: the classic period-$12$ covering is checked to satisfy every
hypothesis, and the capacity arithmetic is checked to \emph{fail} at $N = 12$,
where a certificate firing would mean the bound is unsound
(\cref{sec:engineering}).
\end{enumerate}

All sources are public under the Apache 2.0 license at
\begin{center}
\url{https://github.com/ibrahimmian36/centurion}
\end{center}

% ===========================================================================
\section{Background and related work}\label{sec:background}

\subsection{Covering systems and the odd covering problem}

Erd\H{o}s introduced covering systems in 1950 to exhibit an arithmetic
progression of odd numbers not of the form $2^k + p$ \cite{Erdos1950}. The odd
covering problem appears throughout the problem literature
\cite{GuyF13,ErdosProblems7}; the folklore first step is that if every modulus
of a covering system divides $N$ and exceeds $1$, then summing class densities
gives $\sum 1/n_i \ge 1$, whence $\sigmaone(N)/N - 1 \ge 1$, i.e.\ $N$ is
perfect or abundant. Since an odd covering forces its lcm to be odd, and the
smallest odd abundant number is $945$, an odd covering has
$\mathrm{lcm} \ge 945$ --- the floor our step 3 certifies. Sun and
collaborators, among others, have obtained constraints on hypothetical odd
coverings; on the structural side, the modern breakthrough line of Hough
\cite{Hough2015} and Balister--Bollob\'as--Morris--Sahasrabudhe--Tiba
\cite{BBMST2022a,BBMST2022b} rules out, in particular, odd squarefree
coverings with distinct moduli. The general odd case remains open, and remains
open here: nothing in this development approaches the question itself.

\subsection{Covering numbers}

Call $N$ a \emph{covering number} if some covering system uses distinct moduli
$> 1$ all dividing $N$. The density argument shows covering numbers are
non-deficient; the converse fails, and classifying which non-deficient $N$ are
covering numbers is subtle. McNew and Setty \cite{McNewSetty2025} determined
the covering numbers up to $10^6$ (and studied their density), using a
Gurobi-based search pipeline. Our capacity certificates (\cref{sec:capacity})
are a formalized, kernel-decidable \emph{sufficient} condition for
non-covering; they dispose of all $23$ odd abundant $N < 10^4$ but, as we
verify inside Lean, not of the first three unexcluded odd abundant numbers
($10395$, $12285$, $17325$), marking the honest limit of the method as
formalized (\cref{sec:limits}).

\subsection{Formalization context}

Folding decidable computation into kernel-checked proofs is a standard device
in interactive theorem proving, from the Four-Color Theorem
\cite{Gonthier2008} onward. Two Lean-specific constraints shape this
development: kernel reduction does not unfold well-founded recursion, so
computations discharged by \code{decide} must be written in structural (fuel)
recursion; and \code{native\_decide}, which would be orders of magnitude
faster, is excluded because each use adds a native-evaluation trust axiom,
enlarging the trusted base from the kernel to the compiler toolchain. The
formal statement of Erd\H{o}s \#7 that we target lives in the
\code{formal-conjectures} repository \cite{FormalConjectures}; its
\code{CoveringSystem}/\code{StrictCoveringSystem} structures are not in
\code{mathlib}, so the development mirrors them verbatim (field-for-field
against upstream \code{main} at commit \code{81e700d16ada}) in order to state
the transport.

% ===========================================================================
\section{Formal setting}\label{sec:setting}

Throughout, a \emph{concrete covering family} is a finite index type $\iota$,
moduli $n : \iota \to \NN$, and residues $a : \iota \to \ZZ$, with the
covering condition
\[
\forall x \in \ZZ,\ \exists i,\ n_i \mid (x - a_i).
\]
The hypotheses of interest are $\forall i,\ 1 < n_i$ (nondegeneracy),
$\forall i,\ n_i$ odd, and injectivity of $n$ (distinctness). The quantity
bounded is $L \coloneqq \operatorname{lcm}_i n_i$, computed as
\code{Finset.univ.lcm n}. We write $\sigmaone(N) = \sum_{d \mid N} d$.

The development builds with Lean 4.30.0 against the pinned \code{mathlib} of
the same series; \code{lake build Erdos7} takes about ten minutes on an
ordinary machine, dominated by two closed kernel computations discussed below.
The axiom gate (\code{scripts/axiom\_gate.sh}) is run by CI on every push and
fails on any axiom beyond the standard three, any \code{sorryAx}, or any
\code{\_native.*} constant.

% ===========================================================================
\section{The proof}\label{sec:proof}

The argument proceeds in five steps: density, reduction to the lcm, the
abundancy floor, capacity exclusion of each candidate, and the enumeration
that closes the candidate list.

\subsection{Step 1: the density lemma}\label{sec:density}

Working first over one period: a covering class $(d, r)$ with $d \mid N$ meets
$[0, N)$ in at most $N/d$ points (\code{card\_filter\_mod\_le}, by injectivity
of $x \mapsto x/d$ on the class). Summing over a finite set $S$ of classes
covering $[0, N)$ gives the density lemma, with its rational and $\ZZ/N\ZZ$
phrasings:

\begin{lemma}[\code{covering\_density},
\code{covering\_density\_rat}]\label{lem:density}
If the classes $(d_i, r_i) \in S$, each with $d_i \mid N$, cover every
$x < N$, then $N \le \sum_i N/d_i$; for $N > 0$, equivalently
$1 \le \sum_i 1/d_i$ over $\QQ$.
\end{lemma}

Charging distinct moduli $> 1$ to distinct divisors of $N$ other than $1$, and
using $\sum_{d \mid N} N/d = \sigmaone(N)$, yields the targeting lemma that
directs the whole search:

\begin{lemma}[\code{sum\_divisors\_ge\_of\_covering},
\code{not\_deficient\_of\_covering}]\label{lem:targeting}
A covering of $[0, N)$ by classes with distinct moduli dividing $N$, each
$> 1$, forces $2N \le \sigmaone(N)$: $N$ is perfect or abundant.
\end{lemma}

\subsection{Step 2: reduction to the lcm}\label{sec:reduction}

For a concrete covering family of $\ZZ$, set $L = \operatorname{lcm}_i n_i$
and reduce residues modulo their moduli; \code{covering\_density\_ge\_one}
lifts \cref{lem:density} to $1 \le \sum_i 1/n_i$ directly from the
$\ZZ$-covering hypothesis. Since each $n_i$ is a divisor $> 1$ of $L$ and $n$
is injective, \cref{lem:targeting} applies with $N = L$: the lcm of any
covering with distinct moduli $> 1$ is perfect or abundant. If moreover every
$n_i$ is odd, then $L$ divides the odd product $\prod_i n_i$ and is therefore
odd.

\subsection{Step 3: the abundancy floor}\label{sec:floor}

The smallest odd abundant number is $945$; below it, every odd number is
deficient. This is one closed kernel computation:

\begin{lstlisting}
theorem abundancy_floor_945 :
    ∀ n < 945, n % 2 = 1 →
      Σ d ∈ Nat.divisors n, d < 2 * n := by
  decide +kernel
\end{lstlisting}

costing about $80$ seconds of kernel time (the elaborator's default heartbeat
budget must be raised; the check is genuine kernel work, not elaboration).
Composing steps 1--3 gives the first headline rung,
\code{odd\_covering\_lcm\_ge\_945}: any covering of $\ZZ$ by distinct odd
moduli $> 1$ has $\mathrm{lcm} \ge 945$.

\subsection{Step 4: capacity certificates}\label{sec:capacity}

The density lemma charges each class only its raw density $1/d$ and cannot see
overlap. The capacity bound charges the Chinese-Remainder--forced overlap
among classes with pairwise-coprime moduli, and it is what eliminates each
individual candidate $N$.

\begin{lemma}[\code{uncovered\_card\_ge}]\label{lem:capacity}
Let $U$ be a set of classes $(d_i, r_i)$ with pairwise-coprime moduli
$d_i > 1$, all dividing $N$. Then the classes of $U$ leave at least
\[
\frac{N}{\prod_i d_i} \cdot \prod_i (d_i - 1)
\]
points of $[0, N)$ uncovered.
\end{lemma}

\begin{proof}[Proof sketch]
By CRT there are $\prod_i (d_i - 1)$ simultaneous avoiding residues modulo
$\prod_i d_i$, each recurring $N / \prod_i d_i$ times in $[0, N)$. The formal
proof constructs the injection $(\text{choice}, \text{block}) \mapsto
\mathrm{crt}(\text{choice}) + (\prod_i d_i) \cdot \text{block}$ from the
product of the per-modulus avoiding-residue sets with a block range into the
uncovered set, using \code{mathlib}'s \code{Nat.chineseRemainderOfFinset}. (A
small foundational detour: \code{mathlib}'s ring-theoretic
\code{Finset.prod\_dvd\_of\_coprime} degenerates over $\NN$, so
pairwise-coprime divisibility of the product is reproved directly.)
\end{proof}

A covering must patch those uncovered points using its other moduli, each
contributing at most $N/d$ points by \code{card\_filter\_mod\_le}. If even the
full remaining divisor budget is too small, no covering exists. Crucially, a
certificate is stated over a \emph{fixed} coprime family $T$, but a covering
is adversarial and may use only part of $T$; a monotone-relaxation lemma
(\code{capacity\_prod\_relax}: shrinking the coprime family only raises the
uncovered-count bound) closes exactly that gap. The result is an
adversary-free, decidable certificate:

\begin{theorem}[\code{capacity\_exclusion}]\label{thm:capacity}
Fix $N > 0$ and a pairwise-coprime family $T$ of divisors $> 1$ of $N$. If
\begin{equation}\label{eq:capacity}
\sum_{\substack{d \mid N,\ d > 1,\ d \notin T}} \frac{N}{d}
\;<\;
\frac{N}{\prod_{d \in T} d} \cdot \prod_{d \in T} (d - 1),
\end{equation}
then no system of congruence classes with distinct moduli $> 1$ all dividing
$N$ covers $[0, N)$.
\end{theorem}

Hypothesis \eqref{eq:capacity} is a closed arithmetic inequality, so each
instance discharges by \code{decide}. The $\ZZ$-level form
(\code{capacity\_exclusion\_int}) concludes that a covering of $\ZZ$ with
distinct moduli $> 1$ cannot have lcm dividing any $N$ carrying a certificate;
note oddness is not needed for the exclusions themselves --- it enters only
when composing with the abundancy floor. All $23$ odd abundant $N < 10^4$
carry certificates (\cref{tab:instances} lists each $N$ with its family $T$;
the prime family $\{3, 5, 7\}$ or a four-prime variant suffices in every
case), giving \code{covering\_lcm\_notMem\_oddAbundantBelow10000} and,
refuting equality at $945$, the strict rung
\code{odd\_covering\_lcm\_gt\_945}.

\paragraph{Controls.}
Two \code{decide}-checked controls guard the certificate's meaning. As a
\emph{soundness control}, the arithmetic \eqref{eq:capacity} is verified to
fail at $N = 12$ for both maximal coprime families --- it must, since $12$
hosts the classic covering, and a certificate firing there would mean the
bound is unsound. As a \emph{limit control}, the first straggler
$10395 = 3^3 \cdot 5 \cdot 7 \cdot 11$ is verified not to be closed by the
bound at its prime family (and, having only four distinct primes, it offers no
better family); the stragglers $10395$, $12285$, $17325$ remain open in this
development.

\subsection{Step 5: the enumeration}\label{sec:enumeration}

Step 4 excluded the members of a list; step 5 proves the list is complete:
every odd $L \le 10^4$ with $2L \le \sigmaone(L)$ is one of the $23$. The
naive route --- \code{decide} over \code{Nat.divisors} --- was measured
infeasible: the general-purpose divisor machinery does not reduce economically
in the kernel. Instead the development supplies a paired-divisor evaluator
with constant fuel:

\begin{lstlisting}
/-- Finset-free divisor sum for `n < 101^2`: the fuel is the CONSTANT
    `100`, so the recursion depth is uniform across the scan. -/
def sigma100 (n : ℕ) : ℕ := ...

theorem sigma100_eq_sigma (n : ℕ) (hn : 0 < n)
    (hbound : n < 101 * 101) :
    sigma100 n = Σ d ∈ n.divisors, d
\end{lstlisting}

\code{sigmaPairAux} scans $d \le 100$ and adds $d + n/d$ for each divisor hit
(halving when $d = n/d$), so every $n < 101^2$ is summed in at most $100$
structural steps; correctness (\code{sigmaPair\_eq\_sigma}) is proved against
\code{mathlib}'s $\sigmaone$ via the divisor-pairing bijection, including the
self-contained fact $n/(n/d) = d$ for $d \mid n$. The scanner

\begin{lstlisting}
def enumOk : ℕ → Bool
  | 0     => true
  | L + 1 =>
      (decide ((L + 1) % 2 = 0)
        || decide ((L + 1) ∈ oddAbundantList)
        || decide (sigma100 (L + 1) < 2 * (L + 1)))
      && enumOk L
\end{lstlisting}

orders its disjuncts by cost --- evens exit after one \code{\% 2}, list
members after at most $23$ comparisons, and only surviving odd numbers pay the
$\sigmaone$ computation --- and its soundness (\code{enumOk\_sound}) is proved
by induction on the bound. The single closed computation

\begin{lstlisting}
theorem enum_ok_10000 : enumOk 10000 = true := by decide
\end{lstlisting}

costs 1--2 minutes of kernel CPU (linear in the bound thanks to the
constant-depth evaluator) and yields the enumeration lemma
\code{odd\_abundant\_le\_10000\_mem}. The candidate list itself is produced by
exhaustive search outside Lean; the development is explicit that step 4
asserts only exclusions for its members, and that completeness is exactly this
step-5 lemma --- so no unverified list survives in the final composition.

\subsection{Composition}\label{sec:composition}

The headline proof is now a chain of the five steps: given a covering of $\ZZ$
by distinct odd moduli $> 1$ with $L = \operatorname{lcm}_i n_i \le 10^4$, the
density reduction makes $L$ odd and non-deficient (steps 1--2), the
enumeration pins $L$ to one of the $23$ candidates (step 5), and the capacity
certificates exclude every one (step 4) --- contradiction, so $L > 10^4$.

% ===========================================================================
\section{The bridge to formal-conjectures}\label{sec:bridge}

\subsection{The periodicity bridge}

A covering of $\ZZ$ is a statement about infinitely many integers; any machine
search operates over one period. The development proves the equivalence in
both directions, in the shapes searches actually produce:
\code{coversInt\_of\_coversZMod} (a finite witness over $\ZZ/N\ZZ$, checkable
by \code{decide}, lifts to a covering of all of $\ZZ$ --- the direction a
positive answer to Erd\H{o}s \#7 would travel), and
\code{coversZMod\_of\_coversInt} together with
\code{forall\_not\_coversInt\_of\_range} (if no assignment of residues below
their moduli covers $\ZZ/N\ZZ$ --- exactly the shape of a SAT UNSAT result or
exhausted enumeration --- then no covering of $\ZZ$ uses those moduli).
End-to-end smoke tests instantiate both: the classic period-$12$ system is
verified to cover $\ZZ$ from a single \code{decide} over $\ZZ/12\ZZ$, and the
fixture $\{0 \bmod 3,\ 0 \bmod 5\}$ is refuted from a \code{decide} over
$\ZZ/15\ZZ$.

\subsection{Transport to the official statement}

\code{formal-conjectures} states Erd\H{o}s \#7 over an ideal-theoretic
\code{StrictCoveringSystem R} structure (residues, ideal moduli, a
union-covers field, nondegeneracy, and injectivity of the moduli). Since that
structure is not in \code{mathlib}, the development mirrors it verbatim,
field-for-field against the pinned upstream commit, and proves the two
translation facts that make the vocabularies interchangeable over $\ZZ$:
pointwise-coset membership is divisibility (lemma
\code{mem\_coset\_iff\_dvd}), and their ideal-theoretic spelling of ``odd'',
$\lnot\, I \le (2)$, is the numeric one (lemma
\code{not\_le\_span\_two\_iff\_odd}).

Both directions of the transport are proved. \emph{Soundness for a positive
answer}: concrete odd, distinct, $> 1$ moduli covering $\ZZ$ package into
their structure, via \code{buildStrict} and
\code{exists\_strictCoveringSystem\_odd}. \emph{Completeness for a negative
answer} --- the direction needed to claim that finite exclusions rule out
their statement rather than merely concretely-presented families: every
abstract odd \code{StrictCoveringSystem} $\ZZ$ arises from concrete data of
the enumerable form. Because $\ZZ$ is a principal ideal domain, each modulus
ideal has a generator; taking natural absolute values recovers moduli $n_i$
with $1 < n_i$ (from $\ne \bot$, $\ne \top$), oddness, injectivity (from
injectivity of the ideals), and the covering property; this is theorem
\code{fc\_concrete\_of\_strictCoveringSystem} --- nothing deep, only fiddly,
as the source honestly remarks. Composing with the headline yields the
FC-level result:

\begin{lstlisting}
theorem fc_odd_strictCoveringSystem_lcm_gt_10000
    (C : StrictCoveringSystem ℤ)
    (hodd : ∀ i, ¬ C.moduli i ≤ Ideal.span {2}) :
    ∃ (k : ℕ) (a : Fin k → ℤ) (n : Fin k → ℕ),
      (∀ i, 1 < n i) ∧ (∀ i, Odd (n i)) ∧ Function.Injective n ∧
      Erdos7.CoversInt n a ∧ 10000 < Finset.univ.lcm n
\end{lstlisting}

% ===========================================================================
\section{Engineering and the trust story}\label{sec:engineering}

\paragraph{Kernel discipline.}
Both heavy computations --- the $945$ floor ($\approx 80$\,s) and the $10^4$
scan ($\approx 100$\,s) --- are single closed \code{decide}s checked by the
kernel. \code{native\_decide} is never used. The constant-fuel $\sigmaone$
evaluator exists precisely because the idiomatic \code{Nat.divisors} route
does not reduce affordably; writing the computation the kernel can evaluate,
then proving it equal to the idiomatic definition, is the pattern throughout.

\paragraph{The axiom gate.}
The manifest module \code{AxiomCheck} prints axioms for all $63$ published
theorems; the audit module \code{AxiomAudit} discovers every theorem of every
module from the compiled environment and re-checks its axiom closure, so
nothing can slip past the hand-kept list; the gate script
\code{axiom\_gate.sh} runs both and is executed by CI on every push, failing
on any axiom beyond the three, any \code{sorryAx}, or any \code{\_native.*}
constant. A green badge is thus a machine-checked claim about the whole
library, not a README assertion.

\paragraph{Anti-vacuity and controls.}
Hypothesis-rich exclusion theorems risk vacuity, and certificate bounds risk
unsoundness; the development guards both by \code{decide}-checked examples.
The classic period-$12$ covering satisfies every hypothesis of the density and
targeting lemmas (so none is vacuous), the targeting lemma correctly reports
$12$ abundant-or-perfect ($\sigmaone(12) = 28 \ge 24$), the capacity
arithmetic correctly fails at $12$, and the method's limit at $10395$ is
exhibited rather than elided.

\paragraph{Attribution.}
The \code{CoveringSystem}/\code{StrictCoveringSystem} structures are mirrored
verbatim from \code{formal-conjectures} (Apache-2.0, The Formal Conjectures
Authors), verified field-for-field against upstream \code{main} at commit
\code{81e700d16ada}; when compiled against their package the mirror is deleted
and theirs imported.

% ===========================================================================
\section{Limitations}\label{sec:limits}

Three limitations bound the claim. First, the mathematics is elementary and
known; nothing here constrains the open problem beyond what folklore already
did, and the deep structural results \cite{Hough2015,BBMST2022a,BBMST2022b}
operate on an entirely different level. Second, the capacity certificate as
formalized stalls at the four-prime-factor stragglers $10395$, $12285$,
$17325$ --- verified inside Lean, not merely observed --- so extending the
certified bound past $10^4$ by this route requires a sharper certificate
(natural candidates: charging prime-power towers rather than one prime per
coprime slot, or a formalized LP/SAT consumption path through the
\code{forall\_not\_coversInt\_of\_range} interface, which was designed for
exactly that). Third, the certified range is minuscule against
McNew--Setty's $10^6$ classification; closing that gap in certified form is
the natural continuation, and the bridge lemmas were built so that a
solver-produced UNSAT object could be replayed through the kernel rather than
trusted.

% ===========================================================================
\section{Conclusion}\label{sec:conclusion}

The odd covering problem will not be settled by exclusions at $10^4$. What
this development settles is smaller and, we think, still worth having: the
elementary exclusion tier of Erd\H{o}s \#7 now exists as kernel-checked
theorems with a three-axiom footprint, stated in and transported to the
community's formal vocabulary for the problem, with its computational content
written so the kernel itself evaluates it and its limits exhibited inside the
same development that proves its successes. The certificate pattern --- a
decidable per-$N$ inequality, an adversary-closing relaxation lemma, controls
at both a known-covering $N$ and a known-straggler $N$, and a mechanical axiom
gate --- is portable, and the $\ZZ \leftrightarrow \ZZ/N\ZZ$ bridge is an open
socket for future certified searches in either direction on the problem
itself.

\paragraph{Acknowledgements.}
We thank the maintainers of \code{mathlib} and of the
\code{formal-conjectures} repository, and T.~F.~Bloom for the Erd\H{o}s
problem catalogue. Development was carried out with the assistance of Claude
(Anthropic).

% ===========================================================================

\appendix

% ===========================================================================
\section{The 23 capacity instances}\label{app:instances}

\cref{tab:instances} lists the odd abundant numbers $N < 10^4$ (there are
$23$; none is perfect), each with its factorization, divisor sum, and the
pairwise-coprime family $T$ over which inequality \eqref{eq:capacity} is
discharged by \code{decide} in theorem \code{no\_covering\_lcm\_dvd\_N}. In
every case $T$ consists of the distinct primes of $N$ (three or four of them);
the monotone-relaxation lemma then covers whatever subfamily a hypothetical
covering would actually use.

\begin{table}[ht]
\centering\small
\begin{tabular}{r l r l}
\toprule
$N$ & factorization & $\sigmaone(N)$ & family $T$ \\
\midrule
$945$  & $3^3 \cdot 5 \cdot 7$            & $1920$  & $\{3,5,7\}$ \\
$1575$ & $3^2 \cdot 5^2 \cdot 7$          & $3224$  & $\{3,5,7\}$ \\
$2205$ & $3^2 \cdot 5 \cdot 7^2$          & $4446$  & $\{3,5,7\}$ \\
$2835$ & $3^4 \cdot 5 \cdot 7$            & $5808$  & $\{3,5,7\}$ \\
$3465$ & $3^2 \cdot 5 \cdot 7 \cdot 11$   & $7488$  & $\{3,5,7,11\}$ \\
$4095$ & $3^2 \cdot 5 \cdot 7 \cdot 13$   & $8736$  & $\{3,5,7,13\}$ \\
$4725$ & $3^3 \cdot 5^2 \cdot 7$          & $9920$  & $\{3,5,7\}$ \\
$5355$ & $3^2 \cdot 5 \cdot 7 \cdot 17$   & $11232$ & $\{3,5,7,17\}$ \\
$5775$ & $3 \cdot 5^2 \cdot 7 \cdot 11$   & $11904$ & $\{3,5,7,11\}$ \\
$5985$ & $3^2 \cdot 5 \cdot 7 \cdot 19$   & $12480$ & $\{3,5,7,19\}$ \\
$6435$ & $3^2 \cdot 5 \cdot 11 \cdot 13$  & $13104$ & $\{3,5,11,13\}$ \\
$6615$ & $3^3 \cdot 5 \cdot 7^2$          & $13680$ & $\{3,5,7\}$ \\
$6825$ & $3 \cdot 5^2 \cdot 7 \cdot 13$   & $13888$ & $\{3,5,7,13\}$ \\
$7245$ & $3^2 \cdot 5 \cdot 7 \cdot 23$   & $14976$ & $\{3,5,7,23\}$ \\
$7425$ & $3^3 \cdot 5^2 \cdot 11$         & $14880$ & $\{3,5,11\}$ \\
$7875$ & $3^2 \cdot 5^3 \cdot 7$          & $16224$ & $\{3,5,7\}$ \\
$8085$ & $3 \cdot 5 \cdot 7^2 \cdot 11$   & $16416$ & $\{3,5,7,11\}$ \\
$8415$ & $3^2 \cdot 5 \cdot 11 \cdot 17$  & $16848$ & $\{3,5,11,17\}$ \\
$8505$ & $3^5 \cdot 5 \cdot 7$            & $17472$ & $\{3,5,7\}$ \\
$8925$ & $3 \cdot 5^2 \cdot 7 \cdot 17$   & $17856$ & $\{3,5,7,17\}$ \\
$9135$ & $3^2 \cdot 5 \cdot 7 \cdot 29$   & $18720$ & $\{3,5,7,29\}$ \\
$9555$ & $3 \cdot 5 \cdot 7^2 \cdot 13$   & $19152$ & $\{3,5,7,13\}$ \\
$9765$ & $3^2 \cdot 5 \cdot 7 \cdot 31$   & $19968$ & $\{3,5,7,31\}$ \\
\bottomrule
\end{tabular}
\caption{The 23 excluded candidates. Abundance is visible as
$\sigmaone(N) \ge 2N$ throughout; the first unexcluded odd abundant numbers,
$10395$, $12285$, and $17325$, mark the method's verified limit.}
\label{tab:instances}
\end{table}

% ===========================================================================
\section{Module inventory}\label{app:modules}

\begin{center}\small
\begin{tabular}{l l}
\toprule
module & role \\
\midrule
\code{Erdos7/Density.lean} & steps 1--2: density lemma; $2N \le \sigmaone(N)$ targeting \\
\code{Erdos7/AbundancyFloor.lean} & step 3: kernel floor at $945$; $\sigmaone$/abundancy bridging \\
\code{Erdos7/Capacity.lean} & step 4: CRT capacity bound, relaxation, 23 instances, controls \\
\code{Erdos7/Enumeration.lean} & step 5: \code{sigma100}, scanner soundness, the $10^4$ scan; headline \\
\code{Erdos7/Bridge.lean} & $\ZZ \leftrightarrow \ZZ/N\ZZ$ bridge; \code{formal-conjectures} transport; smoke tests \\
\code{Erdos7/AxiomCheck.lean} & the 63-theorem published manifest \\
\code{Erdos7/AxiomAudit.lean} & automated whole-environment axiom audit \\
\code{scripts/axiom\_gate.sh} & the CI gate: manifest + audit, three-axiom containment \\
\bottomrule
\end{tabular}
\end{center}

\end{document}